\documentclass{llncs}

\usepackage{amssymb,mathrsfs}
\usepackage{amsmath,amscd}
\usepackage[all]{xy}
\usepackage{mik-makro,pdfsync,theorem,graphicx}
\newcommand{\rat}{\mathbb{Q}}

\newcommand{\diff}{\mathrm{diff}}

\newcounter{ourexamplecounter}

\newcommand{\Z}{\mathbb{Z}}

\newcommand{\removed}[1]{}

\newcommand{\eqdef}{\stackrel  {\text{def}} = }
\newcommand{\paragrafik}[1]{\vspace{0.1cm} \noindent {\bf #1.}}

\begin{document}
	\title{Minimization of semilinear automata}
	\author{Mikołaj Bojańczyk \and Sławomir Lasota}
	\institute{University of Warsaw}
	\maketitle

	\begin{abstract}
We investigate finite deterministic automata in sets
with non-homogeneous atoms: integers with successor.
As there are uncountably many deterministic finite automata in this setting,
we restrict our attention to automata with semilinear transition function.
The main results is a minimization procedure for semilinear automata.
The proof is subtle and refers to  
decidability of existential Presburger arithmetic with divisibility predicates.
Interestingly, the minimization is not obtained by the standard partition refinement
procedure, and we demonstrate that this procedure does not necessarily terminate for semilinear automata.
	\end{abstract}

	\section{Introduction}
This paper is a successor of a line of research aiming at studying models of computation in a new set theory,
namy in \emph{sets with atoms} (for motivation and a more detailed exposition see for instance~\cite{BKL11,BL12,BL12nonhomo}).

In set theory, elements of sets are other sets, organized in a well-founded way.
In this paper we work in a new set theory, where one additionally postulates an infinite set of \emph{atoms}.
Then elements of sets are either other sets, or atoms, while atoms themselves have not elements.
Examples of atoms that are often considered are \emph{equality atoms} $(\Nat, =)$, i.e., natural numbers with equality,
or \emph{total order atoms} $(\rat, \leq)$, i.e., rational number with the natural order.
In this paper we focus on \emph{integer atoms} $$(\Int, +1),$$ i.e., the integers with the successor function.

In general, atoms are an algebraic structure over some vocabulary.
The structure is typically assumed to be \emph{homogeneous}~\cite{M11}, i.e., to satisfy the following condition:
every isomorphism between finite substructures extends to an automorphism of the whole structure.
Sets with atoms have particularly good properties when atoms are a relational homogeneous structure
over a finite vocabulary. Examples are the equality atoms or the total order atoms, but not
the integer atoms.\footnote{Integer atoms are homogeneous if one weakens the definition: only isomorphisms
between \emph{finitetely generated} substructures extend to the whole structure.}

Sets with atoms were discovered in the 20ies by Fraenkel, and then investigated by Mostowski and others.
In 90ies, sets with atoms have been rediscovered
in semantics~\cite{GP02,gabbay:fountl}. The paper~\cite{BKL11} rediscovers sets with atoms
in automata theory and observes that one can naturally study different models of computation
in the new set theory. The principal difference is that the notion of finiteness in sets with atoms
is replaced by a more liberal notion of \emph{orbit-finiteness}.
This approach allows to capture in an elegant syntax-independent way some models of automata, for instance
finite memory automata of Francez and Kaminski~\cite{FK94}, or timed automata of Alur and Dill~\cite{AD94}.
The paper~\cite{BKL11} proves also the analog of Myhill-Nerode theorem for orbit-finite automata.
The minimization is effective due to a powerful finite representation theorem provided in the paper.
However, the representation theorem holds only when atoms are a homogeneous relational structure
over a finite vocabulary. 

\paragrafik{Contribution of the paper}
The present paper is a natural attempt to extend the above-mentioned results to non-homogeneous atoms. 
We investigate sets with atoms $(\Int, +1)$
and orbit-finite automata therein. However, the finite representation theorem of~\cite{BKL11}
is not applicable any more, and as one of the consequences the automaton model is far too powerful.
In particular, there is uncountably many non-equivalent orbit-finite automata in sets with the integer atoms,
thus the automata may not be finitely presented as an input to a procedure.
This indicates a necessity of a reasonable restriction of the power of automata.
Our restriction applies to transition relation of an automaton and requires this relation to be semilinear.

It turns out the under the restriction, the deterministic orbit-finite automata with the integer atoms admit a minimization procedure.
This is the main result of this paper.
The proof is a surprisingly subtle and complicated reduction to satisfiability of existential Presburger
arithmetic with divisibility. The latter problem was shown decidable in~\cite{lipshitz}.

\paragrafik{Related research}
The integer atoms exhibit significant similarity to the \emph{timed atoms}, i.e., the structure
\[
(\rat, <, +1).\footnote{
Instead of $(\rat, <, +1)$, equally-well one could consider atoms $(\mathbb{R}, <, +)$.}
\]
In~\cite{BL12} orbit-finite automata with the timed atoms have been shown to subsume timed automata~\cite{AD94}.
Moreover, effective minimization has been shown for a subclass of automata rich enough to subsume timed automata:
automata with timed atoms, with transition relation definable in FO($\rat, <, +1)$ without quantifiers.
This amounts to a more severe restriction than the semilinear restriction considered in this paper.
In fact we suppose that techniques similar to those used in this paper would also apply to semilinear automata
with timed atoms.

Other approaches to minimization of timed automata are discussed in~\cite{ACHDW92,SV96,YL97,T06}.

An extension of the representation theorem to non-homogeneous atoms has been recently formulated proved in~\cite{O12}. 
It applies both to the integer atoms and to the timed atoms. 


\section{Preliminaries}

\paragrafik{Sets with atoms}
The set of atoms is assumed to come equipped with some algebraic structure, like the rationals $\rat$ with the natural order.
The notion of \emph{sets with atoms} makes sense for any algebraic structure of atoms.
We fix in this section an arbitrary such structure, even if later on we 
will stick to the particular structure $(\Int, +1)$, the integers with successor.

The \emph{cumulative hierarchy} of sets is a sequence of sets indexed by ordinal ranks.
At any ordinal, sets of this rank are arbitrary sets whose elements are sets of smaller rank, 
or atoms. In particular, sets of rank $0$ are all subsets of atoms.
We restrict to only those sets in the cumulative hierarchy that are well-behaved in the sense described below.

The intuitive idea is that the given algebraic structure is the only relevant structure of atoms, thus we work
'up to automorphism of atoms'.
Formally, given an automorphism $\pi$ of atoms, i.e., a bijection that preserves the algebraic structure,
$\pi$ may be naturally lifted to any set in the cumulative hierarchy:  apply $\pi$ to all elements,
then to all elements of elements, etc. Application of $\pi$ to a set $X$ we denote by $X \cdot \pi$.
A set $S$ of atoms is said to \emph{support} a set $x$ if every automorphism $\pi$
being identity on $S$ preserves $x$ as well: $X \cdot \pi = X$.  
All atoms trivially support any set. As an example, consider the equality atoms and the set $x$ of 
all atoms except $3$ and $5$, which is supported by $\{3, 5\}$. If we move to the total order atoms,
the set $\{3, 5\}$ supports also the open interval $(3, 5)$.
In the sequel we are interested in finitely-supported sets,
i.e., those supported by some finite set of atoms. 

Now we are prepared to define well-behaved sets: a set in cumulative hierarchy
is \emph{hereditarily finitely-supported} if it is finitely-supported itself, all its elements are finitely-supported, etc.
A hereditarily finitely-supported set that is itself supported by the empty set of atoms we call \emph{equivariant}.
A particular but important special case is that of \emph{equivariant function}, i.e., a function $f : X \to Y$ between
two sets that commutes with any automorphism $\pi$ of atoms:
\[
f( x \cdot \pi) = f(x) \cdot \pi \qquad \text{ for all } x \in X \text{ and all automorphisms } \pi.
\]

The new set theory has a more liberal notion of finiteness than the usual one: orbit-finiteness.
For a set $S$ of atoms, the \emph{$S$-orbit} of a set $x$ is defined as follows:
\[
\{ x \cdot \pi : \pi \text{ is an automorphism of atoms being identity on } S\} .
\]
If $S$ is finite, the above set is called \emph{single-orbit}.
A set is orbit-finite if it is the union of finitely-many single-orbit sets.
Note however that the precise number of orbits of a set may vary depending on the choice of $S$.

Sets with atoms have particularly good properties if the atoms are a \emph{homogeneous} relational structure
over a finite vocabulary~\cite{M11}. Examples are equality atoms, total order atoms, or the random graph considered as a relational structure.
One of these good properties is preservation of orbit-finiteness by the Cartesian product. In consequence,
an $n$-ary relation between orbit-finite sets is orbit-finite as well.

\paragrafik{Sets with integer atoms}
From now on we only consider integer atoms; the automorphism group contains all translations $x \mapsto x + z$, for some $z \in \Int$
and is thus isomorphic to $\Int$ with addition.
This choice of atoms leads to a slightly pathological set theory.
As an example, observe that 
every set is supported by any singleton $\{z\} \subseteq \Int$ as the translation with a fixed point is necessarily
the identity. Thus all sets are in the cumulative hierarchy are hereditarily finitely-supported.
This is why we will mostly work with equivariant sets in the sequel.

The structure $(\Int, +1)$ is not homogeneous as a relational structure; in fact there is no extension of 
the structure by a finite set of relations that would make the structure homogeneous.
It is thus not surprising that integer atoms lead to a set theory that lacks certain important properties.
One of the most important problems is that orbit-finiteness is not preserved by the Cartesian product.
For instance, the atoms $\Int$ is a single-orbit set, but the product $\Int \times \Int$ has infinitely many orbits, 
namely for each $k \in \Int$, the  diagonal $\set{(i,i+k) : i \in \Int}$ is an orbit.
One of the consequences is that there are uncountably many equivariant binary relations in $\Int$.

By an equivariant isomorphism we mean a bijection between two sets that commutes with automorphisms of atoms. 
For $k \geq 0$, let $\Int_k$ denote integers modulo $k$, i.e., the set of equivalence classes of the congruence modulo $k$.
In particular $\Int_0$ is the same as $\Int$ up to equivariant isomorphism.
From the representation theorem of~\cite{O12} it follows that single-orbit sets with integer atoms are particularly simple:
\begin{lemma} \label{lem:single-orbit-sets}
Up to equivariant isomorphism, every one-orbit set with integer atoms is $\Int_k$ for some $k \geq 0$.
\end{lemma}

\paragrafik{Automata with integer atoms}
For any atoms, one can naturally define nondeterministic and deterministic finite automata.
For instance, a NFA consists of:
\begin{itemize}
\item an orbit-finite alphabet $A$,
\item an orbit-finite set $Q$ of states,
\item a transition relation $\delta \subseteq Q \times A \times Q$,
\item subsets $I, F \subseteq Q$ of initial and accepting states.
\end{itemize}
All sets above are implicitly assumed to be hereditarily finitely-supported.
A DFA is obtained by additionally restricting the set $I$ to be a singleton and
the relation $\delta$ to be a function $Q \times A \to Q$.
Note that a NFA is equivariant iff all its components $A$, $Q$, $I$, $F$ and $\delta$ are so.

For many choices of atoms, in particular for homogeneous atoms, the notion of automaton outlined above
is very reasonable. For instance, if the alphabet is just the set of atoms, equivariant NFA with equality atoms
are expressively equivalent to the finite memory automata of Francez and Kaminski~\cite{FK94}
(called also register automata~\cite{DL09}). 
Similarly, in case of total order atoms, equivariant NFA correspond to a variant of finite memory automata with
order testing (cf.~e.g.~\cite{FHL10}).

Unfortunately,  automata with integer atoms are far too powerful, as illustrated in the example below.
\begin{example} \label{ex:worryingex}
Suppose that $K \subseteq \Int$ is any set of integers, e.g.~the positive integers that are prime numbers. Consider the following  language
\begin{eqnarray*}
	\diff(K)= \set{x_1 \ldots x_n \in \Int^* : \mbox{ for all $i \in \set{2,\ldots,n}$, } x_i-x_{i-1}\in K}.
\end{eqnarray*}
We claim that this language is recognized by an equivariant DFA with integer atoms. The automaton has: an initial state $\epsilon$, a sink error state $\bot$, as well as one state $z$ for every $z \in \Int$. It is easy to see that there are three orbits.
The transition function is
\begin{eqnarray*}
	&&\delta(\epsilon, y) = y \\
	&&\delta(z, y) = \begin{cases}
		y  & \mbox{if $y-z \in K$}\\
		\bot & \mbox{otherwise}
	\end{cases} \\
	&&\delta(\bot, y) = \bot
\end{eqnarray*}
Of course, this is not a true finite automaton, because it refers to the set $K$ in its transition relation, and the set $K$ could be anything, e.g.~undecidable. 
Note that there are uncountably many sets $K$, and each one gives a different language.
\end{example}

This is why starting from Section~\ref{sec:sl-sets} on we restrict our attention to \emph{semilinear} automata.


	\section{Semilinear sets with integer atoms}  
\label{sec:sl-sets}

A problem with the integer atoms is that the Cartesian product does not preserve orbit-finiteness. 
To overcome this difficulty we widen our interest from orbit-finite sets to polynomial sets, defined below.

\paragrafik{Polynomial sets}
For any atoms, one may use the name \emph{polynomial sets} for the smallest class of sets
that contains all equivariant orbit-finite sets and is closed under finite products and disjoint unions.
Thus a polynomial set is a finite disjoint union of \emph{monomial sets}, i.e., of finite products of equivariant one-orbit sets.

Under the integer atoms, by Lemma~\ref{lem:single-orbit-sets} every monomial set is, up to equivariant isomorphism, of the form
\begin{align}
\label{eq:monomial}
	\Int_{k_1} \times \Int_{k_2} \times \cdots \times \Int_{k_n} \qquad \mbox{where }k_1,\ldots,k_n \in \Nat.
\end{align}
Without loss of generality we only consider monomials of the form $\Int^k$ or $\Int_k$:
\begin{lemma}\label{lem:only-some-monomials}
Every polynomial set is isomorphic to a finite disjoint union
of monomials of the form $\Z^k$ or $\Z_k$, for $k \in \Nat$.
\end{lemma}
\begin{proof}
Using Lemma~\ref{lem:single-orbit-sets} and the following identities:
\begin{itemize}
\item for $k \geq 1$ there is an equivariant isomorphism between $\Int_k \times \Int$ and the disjoint union of $k$ copies of $\Int$;
\item for $k, l \geq 1$ there is an equivariant isomorphism between $\Int_k \times \Int_l$ and $\Int_m$, where $m$ is the least
common multiplicity of $k$ and $l$.
\end{itemize}
\qed
\end{proof}

\paragrafik{Semilinear sets with integer atoms}
As we have remarked above, the set $\Int^2$, when seen as a set with  integer atoms,  is isomorphic to a disjoint union of a countably infinite number of copies of $\Int$. 
As a consequence, there are uncountably many equivariant subsets of $\Int^2$: just choose any subset of the infinitely many orbits. This means that there is no hope of algorithms working with arbitrary equivariant subsets of the monomial $\Int^2$.
This motivates us to restrict to subsets of polynomial sets that are not just equivariant, but also semilinear, as defined below.

The standard notion of semilinear sets applies to subsets of 
monomials of the form $\Int^k$. For monomials of the form $\Int_k$, 
every subset is considered to be semilinear. 
Then we extend definition to all monomials of the form:
\begin{align*}
	\Int_{k_1} \times \Int_{k_2} \times \cdots \times \Int_{k_n}
        \qquad \mbox{where }k_1,\ldots,k_n \in \Nat;
\end{align*}
a subset of such set is semilinear
if it is semilinear when translated along the isomorphism used in the proof of 
Lemma~\ref{lem:only-some-monomials}.
The property does not depend on the choice of an isomorphism.

%
\begin{definition}[Semilinear sets with integer atoms] \label{def:semilinear}
Consider a subset
\[
R \subseteq X_1 \times \ldots \times X_n
\]
of an arbitrary equivariant monomial set. $R$
is semilinear if for some $n$-tuple of equivariant isomorphisms 
\begin{align*}
(f_i : X_i \to \Int_{k_i})_{i=1\ldots n}, \qquad \mbox{where }k_1,\ldots,k_n \in \Nat,
\end{align*}
the image of $R$ along the isomorphisms is a semilinear subset of the set
\begin{align*}
	\Int_{k_1} \times \Int_{k_2} \times \cdots \times \Int_{k_n}
\end{align*}
in the sense described above.

Finally, a semilinear subset of a polynomial set is defined by choosing a semilinear subset of each of its monomials.
\end{definition}
Note certain delicacy of the above definition.
Observe that $\Int^3$, when interpreted as a set with integer atoms, 
is, similarly like $\Int^2$, isomorphic to a countable disjoint union of copies of $\Int$. 
Therefore there is an equivariant isomorphism between $\Int^2$ and $\Int^3$; 
consider for instance any bijection between orbits. 
However, a semilinear subset of $\Int^2$, translated to
$\Int^3$ via an equivariant isomorphism, is not a semilinear subset of $\Int^3$ in general.

Definition~\ref{def:semilinear} immediately yields the notion of a semilinear $n$-ary relation on orbit-finite sets.
%
 Also as a special case, we get the notion of a semilinear function between polynomial sets.

We are interested in subsets of polynomial sets (which covers the case of relations and functions) which are both semilinear and equivariant.  For instance, there is no semilinear and equivariant bijection between $\Int^2$ and $\Int^3$.

We conclude this section with the following observations useful later:
\begin{lemma} \label{lem:equivariant-orbit-finite-is-semilinear}
Every orbit-finite subset of a polynomial set is semilinear.
\end{lemma}
\begin{corollary} \label{cor:equivariant-function-between-single-orbits}
Every equivariant function between single-orbit sets is semilinear.
\end{corollary}


\section{Semilinear automata and their decision problems}
\label{sec:sl-automata}

Consider the integer atoms.
A NFA is called  \emph{semilinear} if its transition relation is semilinear.

\paragrafik{Representation}
In this section, we study semilinear automata and their decision problems. To speak of  decision problems,  we should first explain how a semilinear automaton is presented as the input to an algorithm.
A finite representation of such automata is easily deducible from our knowledge collected by now.
Basing on Lemma~\ref{lem:single-orbit-sets}, we assume that every orbit is literrally $\Z_k$, and thus 
may be represented by the number $k$. An orbit-finite set is represented as a multiset of orbits.
Transition relation of a semilinear automaton is represented separately for every monomial set.

We start with the observation that we must restrict ourselves to equivariant automata only,
as non-equivariant automata have undecidable emptiness, even in deterministic case:

\begin{theorem}\label{thm:emptiness-undecidable-for-Z-DFA}
	Emptiness is undecidable for semilinear DFA with the integer atoms.
\end{theorem}
\begin{proof}
Recall that every set is finitely supported, namely supported by every singleton subset of atoms.
In consequence, the transition function $\delta : Q \times A \to Q$
of a semilinear automaton may be an arbitrary semilinear function.
We claim that DFAs with semilinear transition function may simulate deterministic two-counter machines.

	Consider a deterministic two-counter machine with zero tests, which has $n$ states. A configuration of this machine can be seen as an element of 
	\begin{align*}
		\set{0,\ldots,n-1} \times \Nat \times \Nat.
	\end{align*}
Denote by $succ$ the function which maps a configuration to its successor. It is well known that the following question is undecidable:
\begin{quote}
	Given $x$ and $y$, decide if there is some $m$ such that $y = succ^m(x)$?
\end{quote}
This question is undecidable even for a fixed machine, and of course also undecidable when the machine is part of the input.
Consider a G\"odel coding of configurations as numbers defined by
\begin{align*}
	f(i,j,k) = n \cdot (2^j \cdot 3^k) + i.
\end{align*}
Under this coding, the successor function is semilinar. More precisely, there is  a semilinear function $g : \Int \to \Int$ such that for every configuration $x=(i,j,k)$,
\begin{align*}
	f(succ(x)) = g(f(x)).
\end{align*}
Therefore, it is  undecidable if there is some $m$ such that $g^m(f(x))=f(y)$. 

Define a semilinear DFA, with a singleton input alphabet, as follows.
Its states are $\Int$, the initial and accepting states are $f(x)$ and $f(y)$, and the transition function
updates state $n$ to $g(n)$ ignoring input.
The automaton is nonempty if and only if the two-counter machine halts.
Thus emptiness of semilinear DFA is undecidable.
\qed
\end{proof}

Somehow surprisingly, equivariance of transition relation makes emptiness easily decidable:

\begin{theorem}\label{thm:emptiness-decidable-for-Z-DFA}
	Emptiness is decidable for semilinear equivariant NFA with the integer atoms.
\end{theorem}
\begin{proof}
	Consider an automaton, with input alphabet $A$, states $Q$, initial states $I\subseteq Q$, accepting states $F$, and transition relation $\delta \subseteq Q \times A \times Q$. Consider the sets
	\begin{align*}
		Q_0  \eqdef I \qquad\mbox{and}\qquad Q_{n} \eqdef Q_{n-1} \cup \delta(Q_{n-1},A) \subseteq Q \mbox{ for }n \ge 1.
	\end{align*}
		It is easy to see that  $Q_n$ is the set of states that  can be reached after reading an input of length at most $n$, and that $Q_n$  can be computed based on $Q_{n-1}$ using Presburger arithmetic. 
Since the set $Q_n$ is an equivariant subset of $Q$,  there are finitely many possibilities for $Q_n$, and therefore the chain  $Q_0 \subseteq Q_1 \subseteq \cdots$ must stabilize at some point.  If it stabilizes without containing an accepting state, the automaton is empty, otherwise the automaton is nonempty.
\qed
\end{proof}

The above theorem makes semilinear equivariant NFA look deceptively simple. The following result illustrates that even in the  deterministic case, semilinear equivariant automata are dangerously close to undecidability. We use below a term \emph{constant word} for any word of the form $x^n$, for some $x \in \Int$ and $n \geq 0$.

\begin{theorem}\label{thm:zero-inputs-undecidable}
	The following problem is undecidable:
	\begin{itemize}
		\item {\bf Input.} A semilinear equivariant DFA with the integer atoms,  with input alphabet $\Z$.
		\item {\bf Question.} Does the automaton accept some constant word? 
	\end{itemize}
\end{theorem}
\begin{proof}
By reduction of the halting problem for two-counter machines.
We will use the same notation as in the proof of Theorem~\ref{thm:emptiness-undecidable-for-Z-DFA}. In particular we will make use of the G\"odel encoding $f$ of configurations of a two-counter machine with zero tests and of the semilinear function $g: \Int \to \Int$
that encodes the transition function of the machine.

Given a machine, define a semilinear equivariant automaton, with input alphabet $\Int$, as follows. Its states are 
\begin{align*}
	\Int \cup \set{\epsilon,\top,\bot},
\end{align*}
where $\epsilon$, $\top$ and $\bot$ have singleton orbits. The initial state is $\epsilon$, and the only accepting state is $\top$. For the initial state, the transition function is defined by
\begin{align*}
	\delta(\epsilon,i) = f(x)+i \in \Int.
\end{align*}
For a state $i \in \Int$, the transition function is defined by
\begin{align*}
	\delta(i,j) = \begin{cases} \top &\mbox{if }g(i-j) = f(y)\\
		j+g(i-j) &\mbox{otherwise}
	\end{cases} 
\end{align*}
The state $\top$ leads to $\bot$ on every input, and $\bot$ is a sink state:
\begin{align*}
	\delta(\top,i) = \bot \qquad \delta(\bot,i) = \bot.
\end{align*}
It is not difficult to see that $\delta$ is equivariant, and that an input of the form $0^m$ is accepted if and only if the two-counter machine goes from configuration $x$ to configuration $y$ in exactly $m$ steps.
An equivariant automaton accepts a constant word if and only if it accepts a word $0^m$ for some $m$.
This completes the proof.
\qed
\end{proof}

\begin{corollary}
	It is undecidable if two  semilinear equivariant DFA have nonempty intersection.
\end{corollary}
\begin{proof}
	The language of all constant words 	is easily seen to be recognized by a semilinear equivariant DFA.
\qed
\end{proof}

Observe how the delicacy of the decidability border between Theorem~\ref{thm:emptiness-decidable-for-Z-DFA} and Theorem~\ref{thm:zero-inputs-undecidable}. Consider a semilinear equivariant DFA with states $Q$. For $n \in \Nat$, define $P_n \subseteq Q$ to be the set of states that can be reached after reading a constant word of length $n$.   
It is not difficult to see that $P_n$ is an equivariant subset of $Q$, and therefore there are finitely many possibilities for $P_n$. However, by Theorem~\ref{thm:zero-inputs-undecidable}, it is undecidable if there is some $n$ such that $P_n$ contains an accepting state.


\section{Minimizing semilinear equivariant automata}
In this section we turn to the problem of minimizing DFAs.
For any choice of atoms, in particular for the integer atoms, every DFA has an equivalent minimal DFA.
We start by proving that when one starts with a semilinear equivariant DFA,  
the minimization operation stays in the realm of semilinear DFAs.
\begin{theorem}
If a language  $L \subseteq A^*$ is recognized by a semilinear equivariant DFA with the integer atoms
then its minimal automaton is also a semilinear equivariant DFA.
\end{theorem}
\begin{proof}
	Let $\Aa$ be a semilinear equivariant DFA recognizing $L$. Let $Q$ be the state space of $\Aa$ and let $P$ be the state space of the syntactic automaton, which is always equivariant. Let $f : Q \to P$ be the mimimizing function, again always equivariant. 
	By Corollary~\ref{cor:equivariant-function-between-single-orbits} we know that $f$ is semilinear.
	Let $\delta : Q \times A \to Q$ be the transition function of the automaton $\Aa$, which is semilinear by assumption. Our goal is to show that the transition function $\gamma: P \times A \to P$ of the syntactic automaton is semilinear. The function $\gamma$ is the same as the relation
	\begin{align*}
		\set{ (f(q),a,f(\delta(q,a))): q \in Q, a \in A},
	\end{align*}
	which can be easily seen to be semilinear, assuming that $\delta$ and $f$ are.
\qed
\end{proof}

The theorem, however, says nothing about computing the minimal automaton. It even says nothing about deciding if an automaton is already minimal. 
\begin{proposition}\label{prop:minim-equiv}
	As far as decidability is concerned, deciding minimality is equivalent to computing the minimal automaton.
\end{proposition}
\begin{proof}
	Suppose that one can decide if an automaton is minimal. We show how to compute the minimal automaton. The algorithm runs two nested loops. 
	
\begin{itemize}
	\item In the outer loop, we input a semilinear automaton $\Aa$ with states.
	First, we test if $\Aa$ is already minimal. If yes, the algorithm terminates and outputs $\Aa$. Otherwise, the algorithm enters the inner loop.
	\item In the inner loop, the algorithm  searches through all equivariant functions  $f: Q \to P$, and tests  each one to see  if it is an automaton homomorphism, which can be expressed in Presburger arithmetic. Under the assumption that $\Aa$ is not minimal, the  inner loop finds some homomorphism $f$ which is not a bijection. Upon finding such a homomorphism,  the inner loop terminates, and the outer loop is executed again.
\end{itemize} 
We claim that the outer loop can only be executed finitely many times. This is because the order 
\begin{align*}
	X > Y  \qquad \mbox{if there is a surjective, non bijective, equivariant function $f : X \to Y$}
\end{align*}
is well founded. The reason is that a non-bijective equivariant  function must either decrease the number of orbits, or decrease the characteristic of some orbit. By a characteristic of $\Int_k$ we mean here the number $k$. Observe that the order does admit arbitrarily long decreasing chains with the same starting point, e.g.
\begin{align*}
	\Int > \Int_{2^k} > \Int_{2^{k-1}} > \cdots > \Int_2 > \Int_1 \qquad \mbox{ for any }k \in \Nat.
\end{align*}
\qed
\end{proof}

\subsection{Partition refinement fails} 
One natural approach to minimization problem would be to use a partition refinement algorithm, described below.

 Consider semilinear equivariant DFA $\Aa$, with states $Q$. For $n \in \Nat$, define an equivalence relation $\sim_n$ on $Q$, which identifies two states if they accept the same inputs of length at most $n$. It is not difficult to see that  $\sim_n$ is a semilinear relation, because its definition can be expressed in Presburger arithmetic. Also,  $\sim_n$ is an equivariant equivalence  relation on $Q$, which means that there is an orbit-finite quotient $Q/_{\sim_n}$. Suppose that  these equivalences stabilize at some $n \in \Nat$, which means that the equivalence relations $\sim_n$ and $\sim_{n+1}$ are the same. Then it is not difficult to prove that $\sim_{n}$ is the Myhill-Nerode equivalence on states of the automaton, and   the minimal automaton has states $Q/_{\sim_n}$.

Unfortunately,  the equivalence might never stabilize:
\begin{proposition}\label{prop:partition-refinement-fails}
	There is a semilinear equivariant DFA, such that equivalences $\sim_n$ never stabilize.
\end{proposition}
Proposition~\ref{prop:partition-refinement-fails} is shown using an automaton such that for every $n \in \Nat$, the set  $Q/_{\sim_n}$ has three orbits, with characteristics $1,1$ and $2^n$ respectively.
\begin{proof}
		Consider an input alphabet
		\begin{align*}
			A = \set{start,\underline 0, \underline 1} \times \Int.
		\end{align*}
		Consider an automaton defined as follows. Its state space is:
		\begin{align*}
			Q = \Int \uplus \set{\epsilon,\bot}.
		\end{align*}
	The initial state is $\epsilon$.	The orbit $\Int$ is accepting, while the orbits $\set \epsilon$ and $\set \bot$ are rejecting. The transition function is defined below, for $i,j \in \Int$. When reading the definition below it is a good idea to look at the case $j=0$.
		\begin{eqnarray*}
			\delta(\epsilon,(\sigma,i))  &=& \begin{cases}
				i &\mbox{when $\sigma = start$}\\
				\bot &\mbox{otherwise}\\
	\end{cases}\\
					\delta(i,(\sigma,j))  &=& \begin{cases}
						\frac{i-j-\sigma}2 + j &\mbox{when $\underline \sigma \in \set{0,1}$ and  $i-j-\sigma$ is even}\\
						\bot &\mbox{otherwise}\\
			\end{cases}\\
					\delta(\bot,(\sigma,j))  &=& \bot
		\end{eqnarray*}
		Consider inputs to the automaton which are of the form
		\begin{align*}
			start(i) \cdot \sigma_1(0) \cdots \sigma_n(0) \qquad \mbox{ for $i \in \Nat$ and  }\sigma_1,\ldots,\sigma_n \in \set{\underline 0,\underline 1}.
			\end{align*}
			It is not difficult to see that such an input is accepted if and only if $\sigma_1 \cdots \sigma_n$ is a prefix of the binary representation of $i$ (the binary representation written with the least significant bit coming first).

		Consider the equivalence relation $\sim_n$ on $Q$ defined by
		\begin{align*}
			q \sim_n p \qquad\mbox{if \ } q,p \in \Int \text{ \ and \ } q = p \mod 2^n.
		\end{align*}
		It is not difficult to see that for every $a \in A$, we have
		\begin{align*}
			q \sim_n p  \qquad\mbox{implies} \qquad \delta(q,a) \sim_{n-1} \delta(p,a).
		\end{align*}
		It follows that states equivalent under $\sim_n$ have the same futures of length at most $n$.  Also, the converse holds: if states are not equivalent under $\sim_n$, then they have different futures of length at most $n$. It follows that the equivalences $\sim_n$ have more and more equivalence classes, as $n$ grows, and never stabilize.
\qed
\end{proof}

\subsection{Computing the minimal automaton}
So, how does one effectively minimize an automaton?  Our approach is to reduce the minimization problem to satisfiability for an extension of Presburger arithmetic, which allows a limited use of the  divisibility predicate. 

A formula of  \emph{existential Presburger arithmetic with divisibility} (EPAD) is a formula of the form
\begin{equation}
	\label{eq:lipszyc}
	\exists x_1 \cdots \exists x_n \  \phi
\end{equation}
where the variables $x_1,\ldots,x_n$ quantify over integers; and
$\phi$ is a quantifier-free formula in the language $(\Int, 0, 1, +, |)$, 
where $|$ stands for divisibility. 
%
%
A typical instance of EPAD is deciding if there is a solution to the system  
\begin{align*}
	& 3x  = 3 \mod y \\
	& 5y = 7 \mod x \\
        & 2x = y - 18
\end{align*}

\begin{theorem}\label{thm:epad}\cite{lipshitz}
	Satisfiability is decidable for EPAD.
\end{theorem}

By a lengthy reduction to Theorem~\ref{thm:epad}, we prove that one can effectively minimize a semilinear equivariant DFA. 
This is the main result of this paper. 

\begin{theorem}\label{thm:decide-semilinear-automaton-minimal}
	Given a semilinear equivariant DFA with the integer atoms, one can compute the minimal automaton.
\end{theorem}
The remaining part of the paper is devoted to the proof of the theorem.


\newcommand{\charm}{\mathrm{char}}

\section{Proof of Theorem~\ref{thm:decide-semilinear-automaton-minimal}}
\label{sec:arithmetic}

After some preparatory lemmas (Lemma~\ref{lem:dimension-reduction} and Lemma~\ref{lem:congr} below)
we formulate Theorem~\ref{thm:congruence-decide} below that essentially says that minimality is decidable
for semilinear equivariant DFAs. By the virtue of Proposition~\ref{prop:minim-equiv}
we know that this is sufficient for proving Theorem~\ref{thm:decide-semilinear-automaton-minimal}.

\subsection{Decidability of minimality}
The following lemma shows that semilinear and equivariant subsets of monomial sets can be  interpreted as semilinear, but not necessarily equivariant, subsets of monomial sets of lower dimensions.

\begin{lemma}\label{lem:dimension-reduction}
	The function
	\begin{align*}
		X \subseteq \Int^k \qquad \mapsto \qquad \set{(i_2,\ldots,i_k) : (0,i_2,\ldots,i_k) \in X}\subseteq \Int^{k-1}
	\end{align*}
	is a bijection between semilinear and equivariant subsets of $\Int^k$, and  semilinear but not necessarily equivariant subsets of $\Int^{k-1}$.
\end{lemma}
\begin{proof}
	It is easy to see that the function produces semilinear sets, because the result of the function can be defined in Presburger arithmetic. The inverse of the function is defined by
	\begin{align*}
		Y\subseteq \Int^{k-1} \qquad \mapsto  \qquad 
		\{ (l, i_1 + l, \ldots, i_{k-1} + l) : (i_1, \ldots, i_{k-1}) \in Y \}
	\end{align*}
\qed
\end{proof}

Consider a semilinear equivariant DFA $\Aa$ with states $Q$, input alphabet $A$, and transition function $\delta$.
In the input alphabet, choose letters 
\begin{align*}
	a_1,\ldots,a_n \in A
\end{align*}
so that every orbit  is represented.  For each of these letters, consider the function
\begin{align*}
	\delta_i : Q \to Q \qquad  \delta_i(q)= \delta(q,a_i).
\end{align*}

\begin{lemma} \label{lem:congr}
	An equivariant equivalence relation $\equiv$ on $Q$ is a congruence in the automaton $\Aa$ if and only if it respects the final states and is a congruence for the functions $\delta_1,\ldots,\delta_n$.
\end{lemma}
\begin{proof}
	Choose any states $p,q \in Q$ such that $p \equiv q$. We need to show that 
	\begin{align*}
		\delta(p,a) \equiv \delta(q,a) \qquad \mbox{for every }a \in A.
	\end{align*}
	Choose then some letter $a \in A$. By choice of $a_1,\ldots,a_n$, there must be some $a_i$ and some permutation $\pi$ such that $a \cdot \pi = a_i$. By equivariance of $\delta$, we have 
	\begin{align*}
		\delta(q,a) 
		\cdot \pi  = \delta(q \cdot \pi, a \cdot \pi)  = \delta_i(q \cdot \pi)
	\end{align*}
	Because $\equiv$ is equivariant, it follows that $p \cdot \pi \equiv q \cdot \pi$, and therefore by the  assumption on $\equiv$, we have
	\begin{align*}
		\delta_i(q \cdot \pi) \equiv \delta_i(p \cdot \pi).
	\end{align*}
	By the same reasoning as above, we have 
	\begin{align*}
		\delta_i(p \cdot \pi ) \equiv \delta(p,a) \cdot \pi.
	\end{align*}
	We have just proved that 
	\begin{align*}
		\delta(q,a) \cdot \pi \equiv \delta(p,a) \cdot \pi.
	\end{align*} By equivariance of $\equiv$, it follows that $\delta(q,a) \equiv \delta(p,a)$.
\qed
\end{proof}

Observe that the functions $\delta_1,\ldots,\delta_n$ are not necessarily equivariant. Actually, they can be completely aribtrary, because every choice of functions $\delta_1,\ldots,\delta_n$ can be lifted to an equivariant transition function 
\begin{align*}
	\delta: Q \times A \to Q
\qquad\mbox{defined by}\qquad
	\delta(q,a) = \delta_i(q) \cdot \pi
\end{align*}
where $a_i$ and  $\pi \in \Int$ are the unique elements satisfying $a = a_i \cdot \pi$. Because we are dealing with a semilinear automaton, all that we know is that $\delta_1,\ldots,\delta_n$  are semilinear.
That is why, in order to prove Theorem~\ref{thm:decide-semilinear-automaton-minimal}, we need to decide if there is an equivalence relation on $Q$ which respects the final states, and which is a nontrivial congruence with respect to the arbitrary semilinear unary operations $\delta_1,\ldots,\delta_n$. This is shown in Theorem~\ref{thm:congruence-decide}, which is the main result of Section~\ref{sec:arithmetic}.

\begin{theorem}\label{thm:congruence-decide}
	The following problem is decidable:
	\begin{itemize}
		\item {\bf Input.} An orbit finite set $Q$, an equivariant subset $F \subseteq Q$,  semilinear but not necessarily equivariant functions $\delta_1,\ldots,\delta_n : Q \to Q$.
		\item {\bf Question.} Is there an equivalence relation $\equiv$ on $Q$ which:
		\begin{enumerate}
			\item is nontrivial, i.e.~identifies at least two different elements;
			\item respects $F$;
			\item is a congruence with respect to $\delta_1,\ldots,\delta_n$;
			\item is equivariant.
		\end{enumerate}
	\end{itemize}
	\end{theorem}


The proof strategy for the theorem is as follows. 
\begin{enumerate}
	\item  We  show that every equivariant equivalence relation $\equiv$ on $Q$ can be  described by a \emph{signature}, which consists of:
	\begin{enumerate}
		\item The \emph{equivalence type}: an equivalence relation $\sim$ on the orbits of $Q$. 
		\item The \emph{equivalence parameters}: a finite vector of integer parameters.
	\end{enumerate}
	This step is done in Section~\ref{sec:the-signature-of-an-equivariant-equivalence-relation}.
\item We show that for every choice of the equivalence type (there are finitely many choices), whether or not the equivalence relation  $\equiv$ is a congruence with respect to $\delta_1,\ldots,\delta_n$ can be described, in terms of the equivalence  parameters,  by an existential formula of Presburger arithmetic with divisibility predicates.
This step is done in Section~\ref{sec:which-signatues-describe-congruences}.
\item We recall Theorem~\ref{thm:epad}, which says that satisfiability is decidable for existential formulas of Presburger arithmetic with divisibility predicates.	
\end{enumerate}
The remaining part of Section~\ref{sec:arithmetic}
is devoted to the proof of Theorem~\ref{thm:congruence-decide}.

\subsection{The signature of an equivariant equivalence relation}
\label{sec:the-signature-of-an-equivariant-equivalence-relation}
In this section, we present the first step of the proof of Theorem~\ref{thm:congruence-decide}. We show how every equivariant equivalence relation on an orbit finite set can be described by a finite piece of information and a vector of numbers.

Consider an orbit-finite equivariant set $Q$, together with an equivariant equivalence relation $\equiv$.
 Define a  relation $[\equiv]$ on orbits of $Q$ by:
\begin{align*}
	\tau_1\ [\equiv]\ \tau_2 \qquad\mbox{if $q_1 \equiv q_2$ for some $q_1 \in \tau_1$ and some $q_2 \in \tau_2$}.
\end{align*}
This relation $[\equiv]$ is called the \emph{equivalence type} of $\equiv$.

\begin{lemma}
	The relation $[\equiv]$ is an equivalence relation.
\end{lemma}
\begin{proof} The only nontrivial part, transitivity, follows from equivariance, as shown below. Suppose that
	\begin{align*}
		\tau_1 \ [\equiv]\ \tau_2 \qquad \mbox{and}\qquad \tau_2 \ [\equiv]\ \tau_3
	\end{align*}
	hold for three orbits $\tau_1,\tau_2,\tau_3$. This means that 
	\begin{align*}
		q_1 \equiv q_2 \qquad \mbox{and}\qquad p_1 \equiv p_2 \qquad \mbox{ for some }q_1 \in \tau_1, q_2,p_2 \in \tau_2, p_2 \in \tau_3.
	\end{align*}
	Sincd  $q_2$ and $p_1$ are in the same orbit, there must be some $\pi$ such that $p_1 \cdot \pi = q_2$. By equivariance of $\equiv$, we have $p_1 \cdot \pi \equiv p_2 \cdot \pi$. Therefore, by transitivity of $\equiv$, we have $q_1 \equiv p_2 \cdot \pi$, and therefore $\tau_1\ [\equiv]\ \tau_3$.
\qed
\end{proof}

\newcommand{\equivquotient}[1]{#1/\!\!\equiv}
We denote by $\equivquotient Q$ the set of equivalence classes of $Q$ under $\equiv$.
Consider the quotient mapping
\begin{align*}
	f : Q \to \equivquotient Q
\end{align*}
which maps an element of $q$ to its equivalence class under $\equiv$, denoted by $[q]_\equiv$. It is not difficult to see that 
$f$ is an equivariant function. As an image of an orbit-finite set under an equivariant function, the quotient $\equivquotient Q$  is orbit-finite. Also, it is not difficult to see that the orbits of $\equivquotient Q$ are in one-to-one correspondence to equivalence classes of the equivalence type $[\equiv]$.



So far we have defined the equivalence type of $\equiv$. This does not yet determine the equivalence relation $\equiv$. The missing information is: what are the characteristics of the orbits in the quotient $\equivquotient Q$, and how does the quotient function identify elements of $Q$.

\begin{itemize}
	\item Let $\Sigma$ be an equivalence class of $[\equiv]$,  corresponding to an orbit  of $\equivquotient Q$. Define $\charm(\Sigma) \in \Nat$  to be the unique number such that the orbit of $\equivquotient Q$ that corresponds to $\Sigma$ is isomorphic to $\Int_{\charm_\Phi(\Sigma)}$.
	\item Let $\Sigma$ be an equivalence class of $[\equiv]$, and let  $\tau,\sigma \in \Sigma$ be orbits of $Q$.  The images $f(\tau)$ and $f(\sigma)$ are equal, and isomorphic to $\Int_{\charm_\Phi(\Sigma)}$. Let $0_\tau, 0_\sigma$ be the copies of $0$ in the  orbits $\tau,\sigma$.  Define 
	\begin{align*}
		\diff(\tau,\sigma) \eqdef f(0_\sigma) - f(0_\tau) \in   \Int_{\charm(\Sigma)}.
	\end{align*}
	In principle, $\diff(\tau,\sigma)$ is an element\footnote{In the definition of $\diff(\tau,\sigma)$, we implicitly use some  isomorphism between $\Int_{\charm(\Sigma)}$ and  the orbit $f(\tau)=f(\sigma)$. There are, however, many possible isomorphisms. It can easily be checked that the difference $f(0_\sigma)-f(0_\tau)$ does not depend on the isomorphism.
	} of $\Int_{\charm(\Sigma)}$. However, it will be convenient to treat it as an integer.  That is why, we think of   $\diff(\tau,\sigma)$ as any integer in the set
	\begin{align*}
		   \diff(\tau,\sigma) + \charm(\Sigma) \cdot \Int.
	\end{align*}
	If there will be a need to make $\diff(\tau,\sigma)$ unique, we can choose it to be the smallest positive integer in the set above.
\end{itemize}


Define  \emph{signature} of the equivariant equivalence relation $\equiv$ to be the triple
\begin{align*}
	\Phi_\equiv \quad\eqdef\quad ([\equiv], \diff, \charm),
\end{align*}
which contains all the information defined above.  This defines a mapping 
\begin{align} \label{eq:trans}
\equiv \ \ \mapsto \ \ \Phi_\equiv .
\end{align}
We claim that the signature determines $\equiv$ uniquely.  To prove this, we 
claim that there is an inverse transformation
\begin{align} \label{eq:inversetrans}
\Phi \ \ \mapsto \ \ \equiv_\Phi
\end{align}
which maps a signature to an equivariant congruence.

Of course, the input to the inverse transformation needs to be consistent. This consistency condition is captured by the notion of an \emph{equivalence signature}. An equivalence signature is  a triple
\begin{align*}
	\Phi = (\sim,\diff,\charm)
\end{align*}
where  $\sim$ is an equivalence relation on orbits of $Q$ and $\diff, \charm$ are vectors
\begin{eqnarray*}
	\diff &:& \set{(\tau,\sigma) \in orbits(Q) \times orbits(Q): \tau \sim \sigma} \to \Int\\
	\charm &:& classes(\sim) \to \Nat.
\end{eqnarray*}	
such that the following consistency conditions hold  for every equivalence class $\Sigma$ of $\sim$:
\begin{itemize}
\item for all $\tau_1,\tau_2,\tau_3 \in \Sigma$:
\begin{eqnarray*}
			\diff(\tau_1,\tau_2) + \diff(\tau_2, \tau_3) = \diff(\tau_1,\tau_3) \mod \charm(\Sigma) .
\end{eqnarray*}	
\item for all $\tau \in \Sigma$, 
\begin{eqnarray*}
k = 0 \mod \charm(\Sigma), 
\end{eqnarray*}
where $\tau$ is isomorphic to $\Int_k$.
\end{itemize}
The first condition above implies, in particular, that 
\begin{align*}
	\diff(\tau,\tau)= 0  \mod \charm(\Sigma)
\qquad\mbox{and}\qquad	\diff(\tau,\sigma) = - \diff(\sigma,\tau)  \mod \charm(\Sigma) .
\end{align*}
In the sequel we will write $\charm(\tau)$ instead of $\charm(\Sigma)$, when $\tau \in \Sigma$.
Clearly, the mapping~\eqref{eq:trans} is surjective onto equivalence signatures.

The following proposition shows that an equivariant equivalence relation on $Q$ is completely described by its signature.
(Note however that there are in general many different signatures defining the same equivalence, because $\diff(\tau,\sigma)$ is only determined modulo 
$\charm(\tau) = \charm(\sigma)$.)

\begin{proposition}\label{prop:equivalence-signatures}
	The exists a mapping~\eqref{eq:inversetrans} that maps equivalence signatures to equivariant equivalences,
	that is an inverse of the mapping~\eqref{eq:trans}, namely:
		\begin{align*}
		\equiv_{\Phi_\equiv} \ \ = \ \ \equiv \qquad \qquad \text{ for any equivariant equivalence } \equiv .
	\end{align*}
\end{proposition}

\subsection{Which signatures describe congruences}
\label{sec:which-signatues-describe-congruences}

\begin{proposition}\label{prop:define-by-equalities}
	Let $\sim$ be an equivalence relation on orbits of $Q$, and let 
\begin{align*}
	f : Q \to Q
\end{align*}
	be a  semilinear function. There set
	\begin{align*}
		X_{f, \sim} = \set{(\diff,\charm) : \mbox{  } (\sim, \diff, \charm) \text{ is an equivalence signature and} \\
		\equiv_{(\sim,\diff,\charm)} \mbox{ is a nontrivial congruence for $f$}}
	\end{align*}
	is definable in EPAD.
\end{proposition}

Before proving Proposition~\ref{prop:define-by-equalities}, we show how it implies Theorem~\ref{thm:congruence-decide}.

\begin{proof}[of Theorem~\ref{thm:congruence-decide}]
		Recall that Theorem~\ref{thm:congruence-decide} says that one can decide if there exists some nontrivial equivariant equivalence relation on $Q$, which respects $F \subseteq Q$, and which is a congruence  with respect to semilinear operations $\delta_1,\ldots,\delta_n$.
	By Proposition~\ref{prop:equivalence-signatures}, every  equivariant equivalence relation on  $Q$ is of the form $\equiv_\Phi$, for some equivalence signature
	\begin{align*}
		\Phi = (\sim,\diff,\charm).
	\end{align*}
 Since the set $F$ is equivariant, it is a union of orbits $F= \tau_1 \cup \cdots \cup \tau_k$.  This means that $\equiv_\Phi$ respects $F$ if and only if the equivalence relation $\sim$ respects the set $\set{\tau_1,\ldots,\tau_k}$.
Recall the sets $X_{f,\sim}$ from Proposition~\ref{prop:define-by-equalities}. There is a nontrivial congruence on $Q$ which respects $F$ if and only if the following set is nonempty
\begin{align*}
	\bigcup_{\text{$\sim$ respects $\tau_1,\ldots,\tau_k$}} \qquad \bigcap_{i \in \set{1,\ldots,n}} \qquad X_{\delta_i, \sim}.
\end{align*}
As EPAD admits positive boolean combinations, it follows from Proposition~\ref{prop:define-by-equalities} that the above set is definable
in EPAD. We can therefore invoke  Theorem~\ref{thm:epad} to test if the set is nonempty.
\qed
\end{proof}

The rest of Section~\ref{sec:which-signatues-describe-congruences} is devoted to proving Proposition~\ref{prop:define-by-equalities}. 
Fix the relation $\sim$ on orbits of $Q$ for the rest of this section.

We will identify an element of $Q$ with a pair $(\tau,i)$, where $\tau$ is an orbit and $i$ is an integer. When the orbit $\tau$ has characteristic $k > 0$, then this representation is many-to-one, because the pairs $(\tau,i)$ and $(\tau,i+k)$ represent the same element.

\begin{lemma}\label{lem:congruence-char}
	Let $\Phi=(\sim,\diff,\charm)$ be an equivalence signature. The  equivalence  $\equiv_\Phi$ is a congruence with respect to a function $f$ if and only if 
	\begin{equation}\label{eq:congruence-needed-1}
	\forall i \in \Int \qquad 	f(\tau,i) \equiv_\Phi f(\tau,i+\charm(\tau)) \qquad \qquad \ \  \mbox{ for every orbit $\tau$}
	\end{equation}
	\begin{equation}\label{eq:congruence-needed-2}
	\forall i \in \Int \qquad 	f(\tau,i) \equiv_\Phi f(\sigma,i+\diff(\tau,\sigma))\qquad \mbox{ for every orbits $\tau \sim \sigma$}
	\end{equation}
\end{lemma}
\begin{proof}
	The left-to-right implication is immediate. We only do the right-to-left implication. 
	A short argument is that the relation $\equiv_\Phi$ is the smallest equivalence relation generated by the pairs
	\begin{align*}
		(\tau,i) \equiv_\Phi (\tau,i+\charm(\tau)) \qquad (\tau,i) \equiv_\Phi (\sigma,i + \diff(\tau,\sigma))
	\end{align*}
	ranging over integers $i \in \Int$ and orbits $\tau \sim \sigma$. Therefore, if $f$ respects the generators, it will also respect the whole equivalence relation.	%
		%
		%
\qed
\end{proof}

The key technical lemma is Lemma~\ref{lem:needed-1}, stated below.
\begin{lemma}\label{lem:needed-1}
	Fix $f$, $\sim$ and orbits $\tau,\sigma$. 
	The set	
	\begin{equation}
		\label{eq:lem-needed}
			X_{f, \sim, \tau, \sigma} \ = \ \{ (\diff, \charm, \Delta) \  : \ \forall i \in \Int \quad 	f(\tau,i) 
			\equiv_{(\sim, \diff, \charm)} f(\sigma,i+\Delta) \}
	\end{equation}
	is definable in EPAD.
\end{lemma}

Before proving the lemma, we show how, together with Lemma~\ref{lem:congruence-char}, it implies Proposition~\ref{prop:define-by-equalities}.
\begin{proof}[of Proposition~\ref{prop:define-by-equalities}]
In order to prove Proposition~\ref{prop:define-by-equalities}, it is sufficient to establish two conditions.
First, we claim that for every choice of $\sim$, the set
\[
\set{(\diff,\charm) : (\sim, \diff, \charm) \text{ is an equivalence signature} }
\]
is definable in EPAD.
Then, by Lemma~\ref{lem:congruence-char}, it suffices to show that for every choice of $\sim$, each of the  properties~\eqref{eq:congruence-needed-1},~\eqref{eq:congruence-needed-2} can be defined in EPAD. This is exactly what Lemma~\ref{lem:needed-1} says. 
\qed
\end{proof}

 The key obstacle for expressing~\eqref{eq:lem-needed} in EPAD is that it has the universal quantifier $\forall i$, which is not allowed in EPAD. The is idea to use a technique of quantifier elimination, which is presented in the following lemma about arithmetic progressions. 
A finite arithmetic progression is a set  $a + p \cdot \set{0,\ldots,k}$. An infinite arithmetic progression is a set $a+ p \cdot \Nat$. Both the base $a$ and the period $p$ can be negative.

\begin{lemma}\label{lem:quantifier-elimination}
	Fix an equivalence signature $\Phi = (\sim, \diff, \charm)$.
	Let $f_1,f_2 : \Int \to Q$ be affine functions defined by
	\begin{align*}
		f_1(x) = (\tau_1,a_1 \cdot x + b_1) \qquad f_2(x) = (\tau_2,a_2 \cdot x + b_2),
	\end{align*}
	for orbits $\tau_1 \sim \tau_2$ and coefficients $a_1,b_1,a_2,b_2 \in \Int$.
	Let $X \subseteq \Int$  be a finite or infinite arithmetic progression. For any  two   consecutive elements $x_0,x_1 \in X$,
	\begin{align*}
		\forall x \in X \quad f_1(x)  \equiv_\Phi f_2(x)\qquad \mbox{iff} \qquad \bigwedge_{x \in \set{x_0,x_1}} f_1(x) \equiv_\Phi f_2(x)
	\end{align*}
\end{lemma}
\begin{proof}
	We only prove the nontrivial right-to-left implication. 
	By unraveling the definition of functions $f_1$ and $f_2$,  $f_1(x) \equiv_\Phi f_2(x)$ is equivalent to 
	\begin{align*}
		\alpha(x) = 0 \mod \charm(\tau_1)
	\end{align*}
	for the function $\alpha : \Int \to \Int$ defined by
	\begin{align*}
		\alpha(x) =  a_1 \cdot x + b_1 - \diff(\tau_1,\tau_2) - a_2 \cdot x - b_2.
	\end{align*}
	This is an affine function. Suppose then that 
	\begin{align*}
		\bigwedge_{x \in \set{x_0,x_1}}\alpha(x) = 0 \mod \charm(\tau_1) 
	\end{align*}
	holds. Then it follows that  
	\begin{align*}
		\forall k \in \Int \qquad \alpha(x_0 + k \cdot(x_1-x_0)) = 0 \mod \charm(\tau_1).
	\end{align*}
	In particular, because $x_1$ and $x_0$ are consecutive elements in the progression $X$, it follows that 
	\begin{align*}
		\forall x \in X \qquad \alpha(x) = 0 \mod \charm(\tau_1),
	\end{align*}
	which means that 
	\begin{align*}
		\forall x \in X \quad f_1(x)  \equiv_\Phi f_2(x)
	\end{align*}
as required.
\qed
\end{proof}

\begin{proof}[of Lemma~\ref{lem:needed-1}]
Consider the functions $g,h : \Int \to Q$ defined by 
\begin{align*}
		g(x) = f(\tau,x) \qquad \mbox{and}\qquad h(x) = f(\sigma,x) .
\end{align*}
As $f$ is semilinear, 
we know that there is a partition of $\Int$ into a finite family of arithmetic progressions $\set{X_i}_{i\in I}$  (respectively, $\set{Y_j}_{j\in J}$), and a family of affine functions $\set{g_i : \Int \to Q}_{i \in I}$  (respectively, $\set{h_j : \Int \to Q}_{j \in J}$), such that on arguments from $X_i$, the functions $g_i$ and $g$ coincide (respectively, on arguments from $Y_j$, the functions $h_j$ and $h$ coincide).

Expressed in terms of these partitions and affine functions,  condition~\eqref{eq:lem-needed} from Lemma~\ref{lem:needed-1} becomes
\begin{align*}
	\forall x  \qquad  \bigwedge_{i \in I, j \in J}  \big(x \in X_i \land (x+\Delta \in Y_j)\big)\quad \Rightarrow \quad g_i(x) \equiv_\Phi h_j(x+\Delta)
\end{align*}
(we write shortly $\Phi$ instead of $(\sim, \diff, \charm)$ here and in the sequel)
which is the same as the conjunction, ranging over all choices of $i \in I$ and $j \in J$, of the properties
\begin{equation}
	\label{eq:conjunct-ij}
	\forall x \in (X_i \cap (Y_j - \Delta)) \quad g_i(x) \equiv_\Phi h_j(x+\Delta)
\end{equation}
Therefore, it is sufficient to provide a formula of EPAD for each formula of the form~\eqref{eq:conjunct-ij}.
Let us fix then $i$ and $j$ in the sequel.

As an intersection of two arithmetic progressions, the set $X_i \cap (Y_j - \Delta)$ is an arithmetic progression. The sets $X_i$ and $Y_j$ are known, only $\Delta$ is unknown. Depending on the value of $\Delta$, the progression $X_i \cap (Y_j - \Delta)$ might have zero, one or at least two elements. 
Thanks to Lemma~\ref{lem:quantifier-elimination}, we know that~\eqref{eq:conjunct-ij} is equivalent to the following property:
\begin{enumerate}
	\item[(A1)] If $X_i \cap (Y_j - \Delta)$ has at least one element $x$, then 
	\begin{align*}
		g_i(x) \equiv_\Phi h_j(x+\Delta).
	\end{align*}
	\item[(A2)] If $X_i \cap (Y_j - \Delta)$ has at least two  elements $x_1,x_2$, and they are consecutive, then 
	\begin{align*}
		\bigwedge_{k \in \set{1,2}} g_i(x_k) \equiv_\Phi h_j(x_k+\Delta).
	\end{align*}
\end{enumerate}
We claim that both these properties can be formalized in EPAD. Let us focus on the more difficult property (A2). It is not difficult to see that the set
\begin{eqnarray*}
		cons&=&\set{ (\Delta,x_1,x_2) : \mbox{ $x_1,x_2$ are consecutive elements of $X_i \cap (Y_j - \Delta$)}}
	\end{eqnarray*}
	are Presburger definable. Property (A2) is an implication. The head of the implication, ``if $X_i \cap (Y_j - \Delta)$ has at least two  elements $x_1,x_2$'' is defined by a Presburger formula  $\exists\ x_1 x_2\ cons(\Delta,x_1,x_2)$. The tail of the implication can be seen as the formula:
	\begin{align*}
		\exists \ x_1 x_2 \qquad cons(\Delta,x_1,x_2) \land \bigwedge_{k \in \set{1,2}} g_i(x_k) \equiv_\Phi h_j(x_k+\Delta),
	\end{align*}
	which is definable in EPAD.
	This completes the proof of Lemma~\ref{lem:needed-1}, being the last part of the proof of Theorem~\ref{thm:congruence-decide}.
	\qed
	\end{proof}

\bibliographystyle{plain} \bibliography{bib}

\end{document}